\documentclass{article}
\usepackage[utf8]{inputenc}

\usepackage[T1]{fontenc}
\usepackage{soul, comment} 
\usepackage{phfqit}
\usepackage{proba} 
\usepackage[sort,nocompress]{cite}
\usepackage{amssymb}
\usepackage{amsmath}
\usepackage{amsthm}
\usepackage{amsfonts}
\usepackage{mathtools}
\usepackage{xspace}
\usepackage{bm}
\usepackage{nicefrac}
\usepackage{commath}
\usepackage{bbm}
\usepackage{boxedminipage}
\usepackage{xparse}
\usepackage{xcolor}
\usepackage{float}
\usepackage{multirow}
\usepackage{makecell}
\usepackage{graphicx}
\usepackage{caption}
\usepackage{subcaption}
\usepackage{ifthen}
\usepackage{thm-restate}
\usepackage{algorithm}
\usepackage[noend]{algpseudocode}
\usepackage{colortbl} 
\usepackage{fancyhdr}   
\usepackage{hyperref}
    \hypersetup{ colorlinks=true, linkcolor=blue, citecolor=magenta, urlcolor=blue,}
\usepackage{fullpage}
\usepackage[utf8]{inputenc}
\usepackage{environ}
\usepackage{mdframed}
\usepackage{cleveref}
\usepackage[short]{optidef} 
\usepackage{enumitem}
\usepackage{tikz}
\usetikzlibrary{shapes}
\usetikzlibrary{positioning}
\usetikzlibrary{fit}
\usetikzlibrary{graphs,graphs.standard}
\usepackage{enumitem}

\newtheorem{proposition}{Proposition}[section]
\newtheorem{lemma}{Lemma}[section]
\newtheorem{remark}{Remark}[section]
\newtheorem{theorem}{Theorem}[section]
\newtheorem{corollary}{Corollary}[section]

\newtheorem{definition}{Definition}[section]
\newtheorem{conjecture}{Conjecture}[section]
\newtheorem{observation}{Observation}[section]

\newtheorem{claim}{Claim}[section]

\newtheorem{question}{Question}

\usepackage{enumitem}

\NewEnviron{problem}[1]{%
	\begin{center}\fbox{\parbox{6in}{%
				{\centering\scshape #1\par}%
				\parskip=1ex
				\everypar{\hangindent=1em}%
				\BODY
}}\end{center}}

\newcommand{\OPT}{\ensuremath{\mathtt{OPT}}\xspace}

\newcommand{\fmax}{\ensuremath{f_{\max}}\xspace}

\newcommand{\indegree}{\ensuremath{d^{\text{in}}}\xspace}

\newcommand{\matroidrank}{\ensuremath{\texttt{rank}}\xspace}

\newcommand{\shubhang}[1]{{\color{blue}[{\tiny Shubhang: \bf #1}]\marginpar{\color{blue}*}}}
\newcommand{\knote}[1]{{\color{red}[{\tiny Karthik: \bf #1}]\marginpar{\color{red}*}}}

\newcommand{\dsg}{\textsc{DSG}\xspace}
\newcommand{\dss}{\textsc{DSS}\xspace}

\newcommand{\mfvs}{\textsc{MatroidFVS}\xspace}

\newcommand{\densitydeletionset}{\textsc{GraphDD}\xspace}
\newcommand{\rhodensitydeletionset}{\ensuremath{\rho\text{-}}\densitydeletionset\xspace}
\newcommand{\dds}{\ensuremath{\textsc{GraphDD}}\xspace}
\newcommand{\rhodds}[1]{\ensuremath{#1\text{-}\dds}\xspace}

\newcommand{\pfds}{\ensuremath{\text{PFDS}}\xspace}

\newcommand{\FVS}{\textsc{FVS}\xspace}
\newcommand{\fvs}{\ensuremath{\text{FVS}}\xspace}

\newcommand{\supmoddensitydeletionset}{\textsc{SupmodDD}\xspace}

\newcommand{\sdds}{\ensuremath{\textsc{SupmodDD}}\xspace}
\newcommand{\rhosdds}[1]{\ensuremath{#1\text{-}\sdds}}

\newcommand{\setcover}{\textsc{SetCover}\xspace}
\newcommand{\submodcover}{\textsc{SubmodCover}\xspace}

\title{On Deleting Vertices to Reduce Density in Graphs and Supermodular Functions\thanks{Grainger College of Engineering, University of Illinois, Urbana-Champaign, Urbana, IL 61801, Email: \{karthe, chekuri, smkulka2\}@illinois.edu. Supported in part by NSF grant CCF-2402667.} }
\author{Karthekeyan Chandrasekaran
\and Chandra Chekuri
\and Shubhang Kulkarni
}
\date{}

\begin{document}
\maketitle
\begin{abstract}
  We consider deletion problems in graphs and supermodular functions where the goal is to reduce density. 
  In Graph Density Deletion (\dds), we are given a graph $G=(V,E)$ with non-negative vertex costs and a non-negative parameter $\rho \ge 0$ and the goal is to remove a minimum cost subset $S$ of vertices such that the densest subgraph in $G-S$ has density at most $\rho$. This problem has an underlying matroidal structure and generalizes several classical problems such as vertex cover, feedback vertex set, and pseudoforest deletion set for appropriately chosen $\rho \le 1$ and all of these classical problems admit a $2$-approximation. In sharp contrast, we prove that for every fixed integer $\rho > 1$, \dds is hard to approximate to within a logarithmic factor via a reduction from \setcover, thus showing a phase transition phenomenon. Next, we investigate a generalization of \dds to monotone supermodular functions, termed Supermodular Density Deletion (\sdds). In \sdds, we are given a monotone supermodular function $f:2^V \rightarrow \mathbb{Z}_{\ge 0}$ via an evaluation oracle with element costs and a non-negative integer $\rho \ge 0$ and the goal is remove a minimum cost subset $S \subseteq V$ such that the densest subset according to $f$ in $V-S$ has density at most $\rho$. We show that \sdds is approximation equivalent to the well-known {\sc Submodular Cover} problem; this implies a tight logarithmic approximation and hardness for \sdds; it also implies a logarithmic approximation for \dds, thus matching our inapproximability bound. Motivated by these hardness results, we design bicriteria approximation algorithms for both \dds and \sdds. 
\end{abstract}

\setcounter{page}{1}
\section{Introduction}\label{sec:introduction}
\label{sec:intro}
The \emph{densest subgraph} problem in graphs (\dsg) is a core primitive in graph and network mining applications. In \dsg, we are given a graph $G=(V,E)$ and the goal is to find $\lambda_G^*:=\max_{S\subseteq V}|E(S)|/|S|$, where $E(S)$ is the set of edges with both end vertices in $S$. 
\dsg is not only interesting for its applications but is a fundamental problem in algorithms and combinatorial optimization with several connections to graph theory, matroids, and submodularity. Many recent works have explored various aspects of \dsg and related problems from both theoretical and practical perspectives \cite{densest-subgraph-survey,boob-20,sw-20,vbk-21,cqt-22,hqc-22,dhulipala2022differential,dinitz2023improved,adaptive-orientations,ne-24}. A useful feature of \dsg is its polynomial-time solvability. This was first seen via a reduction to network flow \cite{g-84,Picard-Queyranne} but another way to see it is by considering a more general problem, namely the \emph{densest supermodular subset} problem (\dss): Given a supermodular function $f:2^V\rightarrow \R_{\ge 0}$ via evaluation oracle, the goal is to find $\lambda^*_f:=\max_{S \subseteq V} {f(S)}/{|S|}$.
One can easily see that \dsg is a special case of \dss by noting that for any graph $G$, the function $f:2^V \rightarrow \mathbb{Z}$ defined by $f(S) = |E(S)|$ for every $S\subseteq V$ is a supermodular function. It is well-known and easy to see that \dss and \dsg can be solved in polynomial-time by a simple reduction to submodular function minimization. Several other problems that are studied in graph and network mining can be seen as special cases of \dss. Recent work has demonstrated the utility of the supermodularity lens in understanding greedy heuristics and approximation algorithms for \dsg and these problems \cite{cqt-22,vbk-21,hqc-22,hqc-23}.

\paragraph{Density Deletion Problems.} In this work we consider several interrelated \emph{vertex deletion} problems that aim to \emph{reduce} the density. 
We start with the graph density deletion problem. 
\begin{definition}[\rhodds{\rho}] For a fixed constant $\rho$, the $\rho$-graph density deletion problem, denoted \rhodds{\rho}, is defined as follows:
\begin{mdframed}
\begin{tabular}{ l l }
 \emph{\textbf{Input}:} &  Graph $G=(V, E)$ and vertex costs $c:V\rightarrow\R_{\ge 0}$ \vspace{1mm}\\  
 \emph{\textbf{Goal}:} & $\arg \min\{\sum_{u\in S}c_S : S\subseteq V \text{ and } \lambda^*_{G-S} \le \rho\}$. 
\end{tabular}
\end{mdframed}
\end{definition}
This deletion problem naturally generalizes to supermodular functions as defined below. We recall that a set function $f: 2^V\rightarrow \R$ is (i) submodular if $f(A)+f(B)\ge f(A\cap B) + f(A\cup B)$ for every $A, B\subseteq V$, (ii) supermodular if $-f$ is submodular, (iii) non-decreasing if $f(A)\le f(B)$ for every $A\subseteq B\subseteq V$, and (iv) normalized if $f(\emptyset)=0$. We observe that non-negative normalized supermodular functions are non-decreasing. For a function $f: 2^V\rightarrow \R$ and $S\subseteq V$, we define $f_{V-S}$ as the function $f$ restricted to the ground set $V-S$. The evaluation oracle for the function takes a subset $S\subseteq V$ as input and returns the function value of the set $S$. 
\begin{definition}[\rhosdds{\rho}] For a fixed constant $\rho$, the $\rho$-supermodular density deletion problem, denoted \rhosdds{\rho}, is defined as follows:
    \begin{mdframed}
\begin{tabular}{ l l }
 \emph{\textbf{Input}:} &  Integer-valued normalized supermodular function $f:2^V\rightarrow\Z_{\ge 0}$ via evaluation oracle and \\
 & element costs $c:V\rightarrow\R_{\ge 0}$ \vspace{1mm}\\  
 \emph{\textbf{Goal}:} & $\arg \min\{\sum_{u\in S}c_u : S\subseteq V \text{ and } \lambda^*_{f_{V-S}} \le \rho\}$. 
\end{tabular}
\end{mdframed}
\end{definition}
When the density threshold $\rho$ is part of input, we use \dds and \sdds to refer to these problems. It is easy to see that \dds (and hence \sdds) is NP-Hard from a general result on vertex deletion problems \cite{Lewis_Yannakakis_1980}. Our goal is to understand the approximability of these problems.

\paragraph{Motivations and Connections.}
While the deletion problems are natural in their formulation, to the best of our knowledge, \dds has only recently been explicitly defined and explored. Bazgan, Nichterlein and Vazquez Alferez \cite{bazgan_et_alSWAT2024} defined and studied this problem from an FPT perspective.  As pointed out in their work, given the importance of DSG and DSS in various applications to detect communities and sub-groups of interest, it is useful to consider the robustness (or sensitivity) of the densest subgraph to the removal of vertices. In this context, we mention the classical work of Cunningham on the attack problem \cite{Cunningham85} which
can be seen as the problem of deleting \emph{edges} to reduce density;
this edge deletion problem can be solved in polynomial time for integer parameters $\rho$ via matroidal and network flow techniques. In addition to their naturalness and the recent work, we are motivated to consider \dds and \sdds owing to their connections to several classical vertex deletion problems as well as a matroidal structure underlying \dds that we articulate next.

We observe that \rhodds{0} is equivalent to the vertex cover problem: requiring density of $0$ after deleting $S$ is equivalent to $S$ being a vertex cover of $G$.  One can also see, in a similar fashion, that \rhodds{1} is equivalent to the pseudoforest deletion set problem, denoted \pfds---where the goal is to delete vertices so that every connected component in the remaining graph has at most one cycle, and \rhodds{(1 - 1/|V|)} is equivalent to the feedback vertex set problem, denoted \fvs---where the goal is to delete vertices so that the remaining graph is acyclic. Vertex cover, \pfds, and \FVS admit $2$-approximations, and moreover this bound cannot be improved under the Unique Games Conjecture (UGC) \cite{KhotR08}. We note that while $2$-approximations for vertex cover are relatively easy, $2$-approximations for \fvs and \pfds are non-obvious \cite{Bafna-Berman-Fujito95,BG96,CHUDAK1998111}. 
Until very recently there was no polynomial-time solvable linear program (LP) that yielded a $2$-approximation for \fvs and \pfds.  
In fact, the new and recent LP formulations \cite{chandrasekaran2024polyhedralaspectsfeedbackvertex} for \fvs and \pfds are obtained via  connections to Charikar's LP-relaxation for \dsg \cite{charikar_greedy_2000}.  Fujito \cite{Fujito-matroid-fvs} unified the $2$-approximations for vertex cover, \fvs, and \pfds via primal-dual algorithms 
by considering a more general class of \emph{matroidal} vertex deletion problem on graphs that is relevant to our work. This abstract problem, denoted \mfvs\footnote{We use the 
\emph{feedback vertex set} terminology in our naming of the \mfvs problem since the goal is to pick a min-cost subset of vertices to cover all \emph{circuits} of the matroid defined on the edges of a graph. This generalizes \fvs which is \mfvs where the matroid of interest is the graphic matroid on the input graph.}, 
is defined below. 
\begin{definition}[\mfvs] The Matroid Feedback Vertex Set problem, denoted \mfvs, is defined as follows:
    \begin{mdframed}
\begin{tabular}{ l l }
 \emph{\textbf{Input}:} &  Graph $G=(V, E)$, vertex costs $c:V\rightarrow\R_{\ge 0}$, and  \\ 
 & Matroid $\calM = (E, \calI)$ with $\calI$ being the collection of independent sets \\
 & \quad \quad (via an independence testing oracle)
 \vspace{1mm}\\  
 \emph{\textbf{Goal}:} & $\arg \min\{\sum_{u\in S}c_u : S\subseteq V \text{ and } E[V-S] \in \calI\}$. 
\end{tabular}
\end{mdframed}
\end{definition}
Fujito \cite{Fujito-matroid-fvs} obtained a $2$-approximation for \mfvs for the class of ``uniformly sparse'' matroids \cite{Lee_Streinu_2008}. 
It is not difficult to show that vertex cover, \fvs, and \pfds can  be cast as special cases of \mfvs 
where the associated matroids are ``uniformly sparse''.  Consequently, Fujito's result unifies the $2$-approximations for these three fundamental problems. 

We now observe some non-trivial connections between $\rho$-\dds, \mfvs and \rhosdds{\rho}. We can show that $\rho$-\dds is a special case of \mfvs for every integer $\rho$: indeed, $\rho$-\dds corresponds to \mfvs where the matroid $\calM_\rho$ is the $\rho$-fold union of the $1$-cycle matroid defined on the edge set of the input graph (see \Cref{thm:dds-to-matroidfvs} in \Cref{appendix:sec:reductions}). Although it is not obvious, we can show that \mfvs is a special case of \rhosdds{1} (see \Cref{thm:matroidfvs-to-supmodDD} in \Cref{appendix:sec:reductions}). 
We refer the reader to the problems in the right column in Figure \ref{fig:reductions}(b) for a pictorial representation of the reductions discussed so far. 
Given these connections and the existence of a $2$-approximation for vertex cover, \fvs, and \pfds, we are led to the following questions. 

\begin{question}\label{question:main}
  What is the approximability of $\rho$-\dds, \mfvs, and $\rho$-\sdds? Do these admit constant factor approximations?
\end{question}

\subsection{Results} 
In this section, we give an overview of our technical results that resolve Question~\ref{question:main} up to a constant factor gap.
\subsubsection{Connections between \submodcover and \sdds}
We obtain a logarithmic approximation for \rhodds{\rho}, \mfvs, and \rhosdds{\rho} via a reduction to the
submodular cover problem and using the Greedy algorithm for it due to Wolsey \cite{Wolsey82}. 
First, we recall the submodular cover problem. 
\begin{definition}[\submodcover] The submodular cover problem, denoted \submodcover, is defined as follows:
    \begin{mdframed}
\begin{tabular}{ l l }
 \emph{\textbf{Input}:} &  Integer-valued normalized non-decreasing submodular function $h:2^V\rightarrow\Z_{\ge 0}$ \\
 & via evaluation oracle and \\ 
    & element costs $c:V\rightarrow\R_{\ge 0}$ \vspace{1mm}\\  
 \emph{\textbf{Goal}:} & $\arg \min\{\sum_{e\in F}c_e : F \subseteq V \text{ and } h(F) \ge h(V)\}$. 
\end{tabular}
\end{mdframed}
\end{definition}

For a function $f: 2^V\rightarrow \R$, we define the marginal $f(v|S):=f(S+v) - f(S)$ for every $v\in V$ and $S\subseteq V$. 
We show the following result. 

\begin{restatable}{theorem}{thmsddsToSubmodCover}\label{thm:sddsToSubmodCover}
Let $f: 2^V\rightarrow \Z_{\ge 0}$ be an integer-valued normalized supermodular function and $\rho$ be a rational number. Then, there exists a normalized non-decreasing submodular function $h:2^V\rightarrow\R_{\geq 0}$ such that 
\begin{enumerate}
    \item if $\rho$ is an integer, then $h$ is integer-valued, 
    \item for $F\subseteq V$, we have that $\lambda^*_{f|_{V-F}} \leq \rho$ if and only if $h(F)\ge h(V)$, 
    \item $h(v) \leq \max\{0, f(v|V - v) - \rho\}$ for all $v \in V$, and 
    \item evaluation queries for the function $h$ can be answered in polynomial time by making polynomial number of  evaluation queries to the function $f$. 
\end{enumerate}
\end{restatable}

We discuss the consequences of Theorem \ref{thm:sddsToSubmodCover} for \rhosdds{\rho}. 
We recall that \submodcover admits a
$(1 + \ln{(\max_{v} h(v))})$-approximation 
for input function $h$ via the Greedy algorithm of Wolsey \cite{Wolsey82}. 
Consider \rhosdds{\rho} for integer-valued $\rho$. 
By Theorem \ref{thm:sddsToSubmodCover}, we have a reduction to \submodcover and consequently, we have a 
$(1 + \ln(\max_{v\in V} f(v|V-v)))$-approximation. 
In particular, we note that \rhodds{\rho} for integer-valued $\rho$ and \mfvs admit $O(\log{n})$-approximation, where $n$ is the number of vertices in the input graph. 

\begin{corollary}\label{coro:rhodds-and-mfvs}
    \rhosdds{\rho} for integer-valued $\rho$ admits an 
    $(1 + \ln{(\max_{v\in V}f(v|V-v))})$-approximation, where $f:V\rightarrow\Z_{\ge 0}$ is the input integer-valued, normalized supermodular function.
    Consequently, \rhodds{\rho} for integer valued $\rho$ and \mfvs admit $O(\log{n})$-approximations, where $n$ is the number of vertices in the input graph. 
\end{corollary}

\begin{remark}
The reduction from \sdds to \submodcover is in some sense implicit in prior literature (see \cite{Ueno_Kajitani_Gotoh_1988, Fujito-matroid-fvs} for certain special cases of supermodular functions).
We note that the reduction from \fvs to \submodcover which follows from this connection does not seem to be well-known in the literature, and the authors of 
this paper were not aware of it until recently.
\end{remark}

From a structural point of view we
also prove that \submodcover reduces to \rhosdds{1}, thus essentially showing the
equivalence of \submodcover and \sdds. We believe that it is useful to have this equivalence explicitly known given that vertex deletion problems arise naturally but seem different from covering problems on first glance.

\begin{restatable}{theorem}{thmSubmodCovertosdds}\label{thm:SubmodCovertosdds}
Let $h: 2^V\rightarrow \Z_{\ge 0}$ be an integer-valued normalized non-decreasing submodular function. Then, there exists a normalized supermodular function $f:2^V\rightarrow \Z_{\ge 0}$ 
such that 
\begin{enumerate}
    \item for $F\subseteq V$, we have that $h(F)\ge h(V)$ if and only if $\lambda^*_{f|_{V-F}}\le 1$, 
    \item 
    $f(v|V-v) = h(v) + 1$ for all $v \in V$, and 
    \item evaluation queries for the function $f$ can be answered in polynomial time by making a   constant number of evaluation queries to the function $h$. 
\end{enumerate}
\end{restatable}

\subsubsection{Hardness of Approximation}
A starting point for our attempt to answer Question~\ref{question:main} was
our belief that $\rho$-\dds for integer $\rho$ admits a $(\rho+1)$-approximation
via the primal-dual approach suggested by Fujito for \mfvs \cite{Fujito-matroid-fvs}. This belief stems from Fujito's work which showed a $2$-approximation for vertex cover, \fvs, \pfds, and \mfvs for  ``uniformly sparse'' matroids and our reduction showing that $\rho$-\dds for integral $\rho$ is a special case of \mfvs (see Theorem \ref{thm:dds-to-matroidfvs}). We note that the matroid that arises in the reduction is not a ``uniformly sparse'' matroid but has lot of similarities with it, so our initial belief was that a more careful analysis would lead to a constant factor approximation. 
However, to our surprise, after several unsuccessful attempts to prove a constant factor
upper bound, we were able to show that for every integer $\rho \geq 2$,  $\rho$-\dds is $\Omega(\log n)$-hard to approximate via a reduction from Set Cover. 

\begin{restatable}{theorem}{thmDDSlognHard}\label{thm:DDS-logn-hard}
For every integer $\rho \geq 2$,  there is no $o(\log n)$ approximation for \rhodds{\rho} assuming $P\neq NP$, where $n$ is the number of vertices in the input instance.
\end{restatable}

Thus, $\rho$-\dds exhibits a \emph{phase transition}: it admits a $2$-approximation for $\rho\le 1$ (via Fujito's results \cite{Fujito-matroid-fvs}) and becomes $\Omega(\log{n})$-hard for every integer $\rho\ge 2$. To conclude our hardness results, we note that since \dds is a special case of  \mfvs, which itself is a special case of \sdds, both \mfvs and \sdds are $\Omega(\log n)$-inapproximable. However, both these problems are also $O(\log{n})$-approximable via Corollary \ref{coro:rhodds-and-mfvs}. Thus, we resolve the approximability of all these problems to within a small constant factor. We refer the reader to Figure \ref{fig:reductions} for an illustration of problems considered in this work and approximation-factor preserving reductions between them.

\begin{figure}[ht]
\centering
\includegraphics[width=0.8\textwidth, trim={2cm 1cm 4cm 2.5cm}, clip]{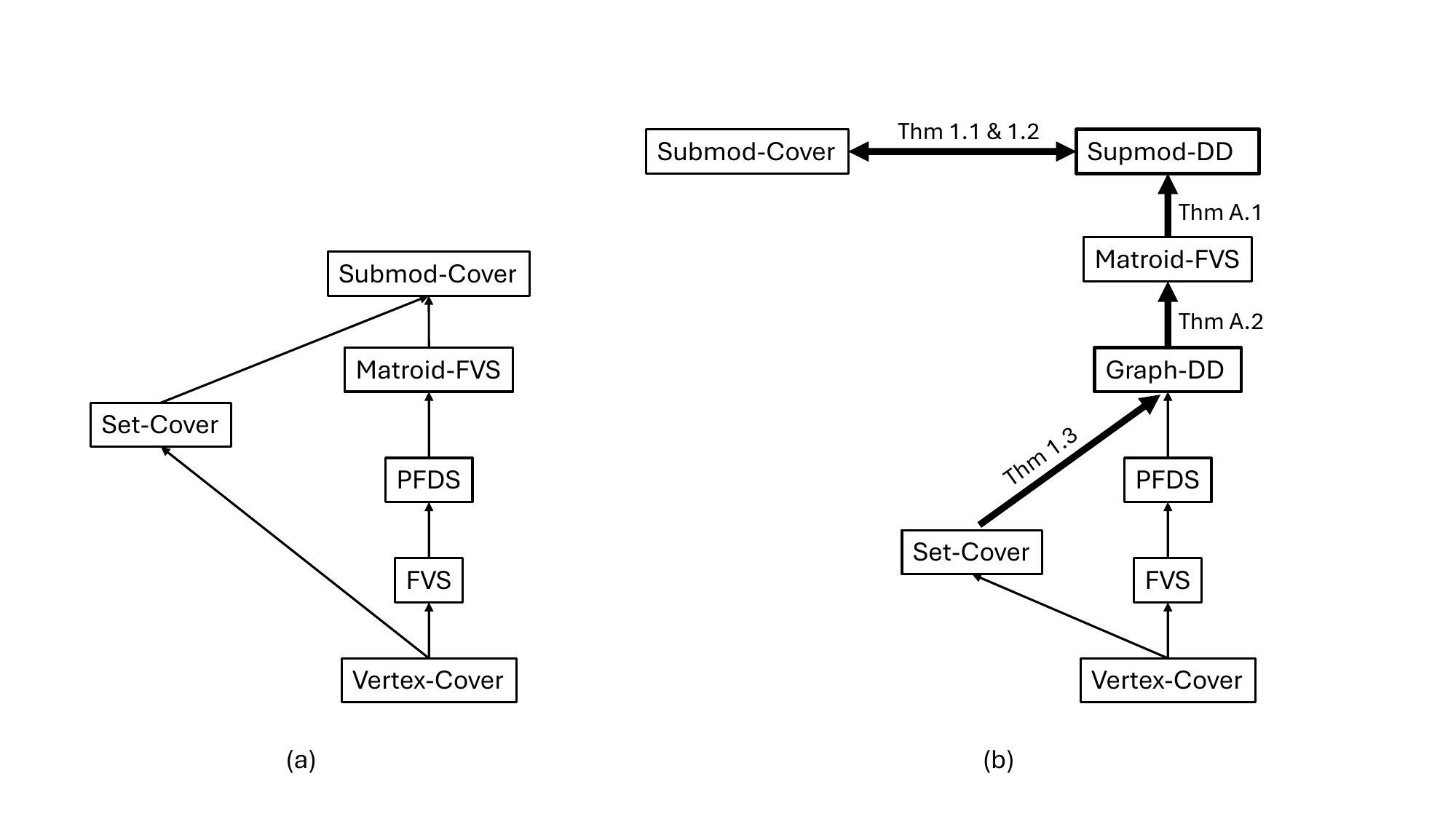}
\caption{Reductions between problems of interest to this work. Arrow from Problem A to Problem B implies that Problem A has an approximation-preserving reduction to Problem B. Figure (a) consists of the connections between problems known prior to our work. Figure (b) showcases our results.}
\label{fig:reductions}
\end{figure}

\subsubsection{Bicriteria Approximations}
The hardness result for $2$-\dds (and $\rho$-\dds) motivates us to consider bicriteria approximation algorithms. Can we obtain constant factor approximation by relaxing the requirement of meeting the density target $\rho$ exactly? We show that this is indeed possible.
We consider an orientation based  LP that was used recently to obtain a polynomial-time solvable LP to approximate \fvs and \pfds \cite{chandrasekaran2024polyhedralaspectsfeedbackvertex}.
We observed that this LP has an $\Omega(n)$ integrality gap when considering $2$-\dds. Nevertheless, the LP is useful in obtaining the following bicriteria approximation.

\begin{restatable}{theorem}{thmbicriteriaorientation}\label{thm:orientation-LP-bicriteria}
    There exists a polynomial time algorithm which takes as input a graph $G = (V, E)$, vertex deletion costs $c:V\rightarrow\R_{\ge 0}$, a target density $\rho\in \R$, and an error parameter $\epsilon \in (0, 1/2)$, and returns a set $S\subseteq V$ such that:
    \begin{enumerate}
        \item $\lambda^*_{G - S} \leq \left(\frac{1}{1 - 2\epsilon}\right) \cdot \rho$,
        \item $\sum_{u\in S}c_u \leq \left(\frac{1}{\epsilon}\right) \cdot \mathtt{OPT}$,
    \end{enumerate}
    where \texttt{OPT} denotes the cost of an optimum solution to \rhodensitydeletionset on the instance $(G,c)$.
\end{restatable}

Next, we consider $\rho$-\sdds. Unlike the case of graphs, it is not clear how to write an integer programming formulation for \sdds whose LP-relaxation is polynomial-time solvable. Instead, 
we take inspiration from the very recent work of \cite{Włodarczyk_2024} on vertex deletion to reduce treewidth. We design a
combinatorial randomized algorithm that yields a bicriteria approximation for \sdds, where the bicriteria approximation bounds are based on a parameter $c_f$ that depends on the input supermodular function $f$. 
For a normalized supermodular function $f:2^V\rightarrow\R_{\ge 0}$, we define 
        \[c_f := \max\left\{\frac{\sum_{u \in S}f(u|S-u)}{f(S)}: S\subseteq V\right\}.
        \]
This parameter $c_f$ was defined in a recent work 
on \dss to unify the analysis of the greedy peeling algorithm for \dsg \cite{cqt-22}. 
We note that $1 \leq c_f \leq |V|$ and moreover, $c_f = 1$ if and only if the function $f$ is modular. If $f$ is the induced edge function of a graph (i.e., $f(S)$ is the number of edges with all its end-vertices in $S$ for every subset $S$ of vertices), then $c_f \le 2$. This follows from the observation that the sum of degrees is at most twice the number of edges in a graph. Similarly, if $f$ is the induced edge function of a hypergraph with rank $r$ (i.e., all hyperedges have size at most $r$), then $c_f \leq r$. 
We show the following bicriteria approximation for \sdds. 

\begin{restatable}{theorem}{thmbicriteriarandomdeletion}\label{thm:bicriteria-random-deletion}
    There exists a randomized polynomial time algorithm which takes as input a normalized  monotone supermodular function $f:2^V\rightarrow\Z$ (given by oracle access), element deletion costs $c:V\rightarrow\R_{\ge 0}$, a target density $\rho\in \R$, and an error parameter $\epsilon \in (0, 1)$, and returns a set $S\subseteq V$ such that:
    \begin{enumerate}
        \item $\lambda^*_{f|_{V - S}} \leq c_f(1+\epsilon)\cdot \rho$,  and 
        \item $\E\left[\sum_{u\in S}c_u\right] \leq c_f \left(1+\frac{1}{\epsilon}\right)\cdot \mathtt{OPT}$,
    \end{enumerate}
    where \texttt{OPT} denotes the cost of an optimum solution to \rhosdds{\rho} on the instance $(f,c)$.
\end{restatable}

As a consequence of Theorem \ref{thm:bicriteria-random-deletion}, we obtain a bicriteria approximation for density deletion problems in graphs and $r$-rank hypergraphs. We note that the bicriteria guarantee that we get from this theorem for graphs is weaker than the guarantee stated in Theorem \ref{thm:orientation-LP-bicriteria}. 
We discuss another special case of \sdds where the supermodular function of interest has bounded $c_f$ to illustrate the significance of Theorem \ref{thm:bicriteria-random-deletion}. Given a graph $G=(V, E)$ and a parameter $p\ge 1$, the $p$-mean density of $G$ is defined as $\max\{\sum_{u\in S}d_S(u)^p/|S|: S\subseteq V\}$, where $d_S(u)$ is the number of edges in $E[S]$ incident to the vertex $u$. The $p$-mean density of graphs was introduced and studied by Veldt, Benson, and Kleinberg \cite{vbk-21}. Subsequent work by Chekuri, Quanrud, and Torres \cite{cqt-22} showed that $p$-mean density is a special case of the densest supermodular set problem (i.e., \dss) where the supermodular function $f: 2^V\rightarrow \R_{\ge 0}$ of interest is given by $f(S):=\sum_{u\in S}d_S(u)^p$ for every $S\subseteq V$ and moreover $c_f\le (p+1)^p$. We note that the natural vertex deletion problem is the $p$-mean density deletion problem: 
given a graph $G=(V, E)$ with vertex deletion costs, 
find a min-cost subset of vertices to delete so that the $p$-mean density of the remaining graph is at most a given threshold. Since $c_f\le (p+1)^p$ for the function $f$ of interest here, Theorem \ref{thm:bicriteria-random-deletion} implies a bicriteria approximation for $p$-mean density deletion  for integer-valued $p$. 
An interesting open question is to obtain better bicriteria approximation for \rhosdds{\rho}---in particular, can we remove the dependence on $c_f$? 

\paragraph{Organization.} The main body of the paper is organized as follows. In \Cref{sec:preliminaries}, we give preliminaries which will be used throughout the technical sections. In \Cref{sec:set-cover-hardness}, we show an approximation-preserving reduction from \setcover to \rhodds{\rho} and prove \Cref{thm:DDS-logn-hard}. In \Cref{sec:bicriteria-algorithm}, we give bicriteria approximations for 
\rhodds{\rho} and \rhosdds{\rho}. Here, we prove \Cref{thm:orientation-LP-bicriteria} in \Cref{sec:orientation-bicriteria} and \Cref{thm:bicriteria-random-deletion} in \Cref{sec:supmodDD-bicriteria}. Finally, in \Cref{sec:submodcover-reductions} we show the connections between \sdds and \textsc{Submodular Cover} by proving Theorems \ref{thm:sddsToSubmodCover} and \ref{thm:SubmodCovertosdds}.

\subsection{Preliminaries: Characterizing Density Using Orientations}\label{sec:preliminaries}
In this section, we give a characterization of density via \emph{(fractional) orientations} which we use throughout the technical sections. We recall that an \emph{orientation} of a graph $G=(V,E)$
assigns each edge $\{u,v\}$ to either $u$ or $v$. A \emph{fractional orientation} is an assignment of
two non-negative numbers $z_{e,u}$ and $z_{e,v}$ for each edge $e=\{u,v\} \in E$ such that 
$z_{e, u}+ z_{e,v}=1$. We note that an orientation is a fractional orientation where all values in the assignment are either $0$ or $1$. For notational convenience, we use $\vec{G} = (V, \vec{E})$ to denote an orientation of the graph $G$ and $\indegree_{\vec{G}}(u)$ to denote the indegree of a vertex $u \in V$ in the oriented graph $\vec{G}$.

The following connection between density and (fractional) orientations will be used in our hardness reduction in \Cref{sec:set-cover-hardness}, our ILP in  \Cref{sec:orientation-bicriteria}, and the connection between \rhodensitydeletionset and \mfvs in \Cref{appendix:sec:reductions}. 
\begin{proposition}\label{thm:orientation-characterization}
Let $G=(V, E)$ be a graph. Then, we have the following:
\begin{enumerate}
    \item let $\rho \in \R_{\geq 0}$ be a real value; then, $\lambda^*_G\le \rho$ if and only if there exists a fractional orientation $z$ such that $\sum_{e\in \delta(u)}z_{e,u}\le \rho$ for every $u\in V$, and
    \item let $\rho \in \Z_{\geq 0}$ be an integer value; then, $\lambda^*_G\le \rho$ if and only if there exists an orientation $\vec{G} = (V, \vec{E})$ of the graph $G$ such that $\indegree_{\vec{G}}(u) \leq \rho$ for every $u\in V$.
\end{enumerate}
\end{proposition}
The first part of the proposition states that a graph $G$ has density at most $\rho$
if and only if the edges of $G$ can be fractionally oriented such that the total fractional in-degree at every vertex is at most $\rho$. This characterization is implied by the dual of Charikar's LP \cite{charikar_greedy_2000} to solve \dsg.
The second part of the proposition is a strengthening of the first part when $\rho$ is an integer and states that a graph $G$ has density at most $\rho$ iff the edges of $G$ can be (integrally) oriented such that the in-degree at every vertex is at most $\rho$.  This characterization is implied by a result of Hakimi \cite{Frank-book} and also via the dual of Charikar's LP \cite{charikar_greedy_2000}.  

\section{Approximation Hardness}\label{sec:set-cover-hardness}


In this section, we show \Cref{thm:DDS-logn-hard}, i.e., we show an approximation preserving reduction from \setcover to \rhodds{\rho}. We recall the set cover problem and its inapproximability. 

\begin{definition}[\setcover]
    The set cover problem, denoted \setcover, is defined as follows:
     \begin{mdframed}
\begin{tabular}{ l l }
 \emph{\textbf{Input}:} &  Finite Universe $\calU$, Family $\calS \subseteq 2^\calU$ with costs $c: \calS \rightarrow \R_{\ge 0}$
    \vspace{1mm}\\  
 \emph{\textbf{Goal}:} & $\arg \min\{\sum_{S\in \mathcal{F}}c_e : \mathcal{F} \subseteq \calS \text{ with } \cup_{S \in \mathcal{F}}S = \calU\}$. 
\end{tabular}
\end{mdframed}
\end{definition}

\begin{theorem}\cite{feige-98,DS14}\label{thm:set-cover-hardness}
    For every $\epsilon>0$, there does not exist a $(1-\epsilon)\ln{n}$-approximation for \setcover assuming P $\neq$ NP, where $n$ is the size of the input instance. 
\end{theorem}

\paragraph{Warmup. } As a warmup, we describe the idea underlying the reduction to prove \Cref{thm:DDS-logn-hard} by giving the proof of a weaker hardness result---we will show that there is no $(\rho+1-\epsilon)$-approximation for \rhodensitydeletionset unless $P = NP$. In particular, given a set cover instance $(\calS, \calU, c)$, we will construct a graph $H = (V_H, E_H)$ with vertex deletion costs $c_H:V\rightarrow\R$ such that a $(\fmax-1)$-\densitydeletionset of $H$ corresponds exactly to a set cover of $(\calS, \calU)$, where $f_{\max}$ is the maximum frequency of an element. Since there is no $(f_{\max} - \epsilon)$-approximation for \setcover, we obtain the claimed result. For simplicity (and without loss of generality), we henceforth assume that every element has the same frequency $f_{\max}$. Our reduction proceeds by considering the incidence graph of the set cover instance---this is the bipartite graph where there is a vertex $v_S$ for each set $S \in \calS$, vertex $u_e$ for each element $e \in \calU$, and edges $(v_S, u_e)$ for each $S\in\calS$ and $e \in \calU$ such that $e \in S$. We make a simple modification to this incidence graph to obtain the graph $H = (V_H, E_H)$: for every set $S\in \calS$, we add $f_{\max}-1$ self loops to the vertex $v_S$. We also define the the cost function as $c_H(v_S) = c(S)$ for all $S \in \calS$ and $c(u_e) = \infty$ for all $e\in \calU$.

\begin{remark}
    The addition of self-loops is not technically necessary for the construction; they simply make it more straightforward. Specifically, a vertex \( u \) with \( \gamma \in \mathbb{Q}_{\geq 1} \) self-loops can be replaced by a subgraph with density exactly \( \gamma \). In this subgraph,  one vertex, say \( h_u \), is identified with $u$ and its cost is defined as \( c(h_u) := c(u) \). All other vertices have infinite cost. All edges incident to \( u \) are then redirected to connect to \( h_u \) instead.
\end{remark}

The intuition behind the addition of self-loops is that the subgraph induced by every element-vertex and its neighborhood (which contains exactly $f_{\max}$ set-vertices) is strictly larger than $\rho$. This is because the number of edges here is $\fmax (1 + \fmax)$ whereas the number of vertices is $(1 + \fmax)$. Consequently, any $(\fmax-1)$-\densitydeletionset must delete at least one vertex from this induced subgraph. Our cost function ensures that the deleted vertices are set-vertices, and so these correspond to a set cover. We argue the reverse direction of the reduction by leveraging the connection between integer density and integral orientations. We now show the reverse direction. Let $F \subseteq \calS$ be a set cover. Let $X_F = \{v_S : S \in F\}$. Then, we show that the graph $H - X_F$ has density at most $\rho$ by exhibiting a simple orientation in which the indegree of every vertex is at most $\rho$: for every remaining edge of type $(v_S, u_e)$, orient the edge from $v_S$ to $u_e$. Orient all self-loops arbitrarily. We note that the indegree of every vertex $v_S$ is at most $\rho$ since there are only $\rho$ self loops by construction. By way of contradiction, suppose that the indegree of a vertex $u_e$ was strictly larger than $\rho$. We note that by construction, $d_H(u_e) = \rho+1$. Consequently, no neighbor of $u_e$ is in $X_F$. Thus, $F$ is not a set cover, a contradiction.

Our proof of \Cref{thm:DDS-logn-hard} follows the same high-level strategy: obtain an appropriate modification of the incidence graph so that for every element, there is an appropriate subgraph with density larger than the target density. However, in contrast to the situation in the reduction we described above, \Cref{thm:DDS-logn-hard} is a statement about all \emph{fixed} integer target densities at least $2$ and so we cannot set the target density as a parameter in our reduction. This leads to additional complications which we overcome by replacing the vertices and edges of the incidence graph with appropriate (binary tree) gadgets. 
We now restate and prove \Cref{thm:DDS-logn-hard}.

\thmDDSlognHard*
\begin{proof}
    We show the theorem when the target density $\rho = 2$ via a reduction from \setcover. At the end of the proof, we remark on how to modify the reduction to obtain hardness for all integers $\rho \geq 2$ as claimed in the theorem. Let $(\calS, \calU, c:\calS\rightarrow\R)$ be a \setcover instance. We will assume (without loss of generality) that all elements in $\calU$ have the same frequency $f \geq 4$ which is a power of $2$---we note that this assumption is not a technical requirement and is only for ease of exposition. For this instance, we construct an instance $(G = (V, E), c_G:V\rightarrow\R)$ of $2$-\densitydeletionset as follows:
    \begin{enumerate}
         \item \emph{Add vertices representing sets:} For each set $S\in \calS$, add a set-vertex $v_S$ to $V$.
         
        \item \emph{Add binary trees representing elements:}
        For each element $e\in \calU$, add a complete binary tree with the $f$ leaves as the set-vertices corresponding to the sets containing $e$. We denote this tree as $\calT_e$ and its root as $r_e$.


        \item \emph{Add self-loops:}
        For each element $e \in E$, add a self-loop to the root vertex $r_e$ of the tree $\calT_e$. For each set $S \in \calS$, add $\rho=2$ self-loops to the vertex $v_S$.

        \item \emph{Define cost function:} We define the cost function $c_G:V\rightarrow\R$ as follows: $c_G(v_S) = c(S)$ for all $S \in \calS$ and $c_G(u) = \infty$ for all $u \in V - \{v_S : S\in \calS\}$.
    \end{enumerate}

\begin{figure}
    \centering
    \includegraphics[width=\textwidth, trim={0cm 3cm 0cm 3cm}, clip]{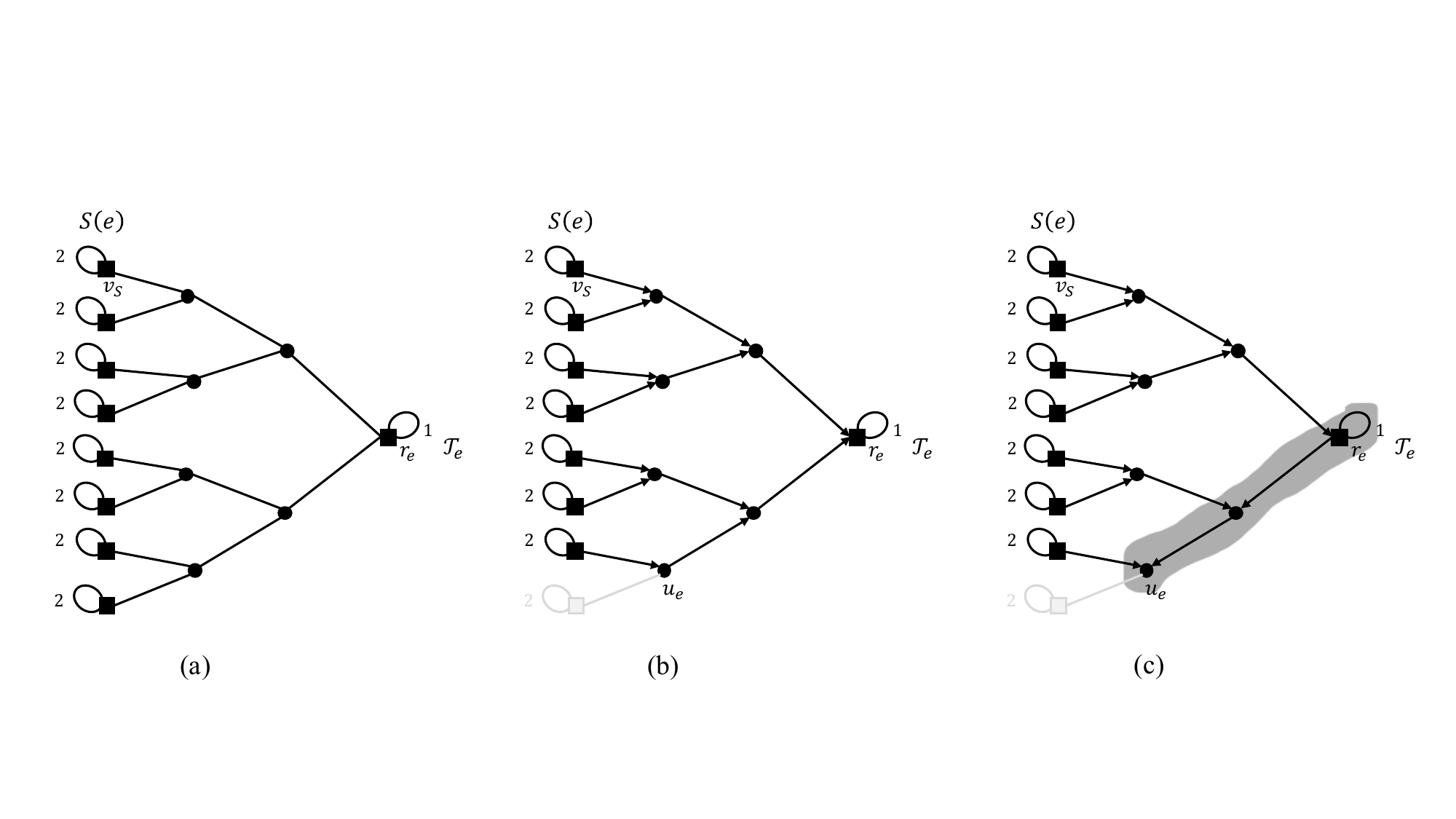}
    \caption{The figure in (a) depicts the subgraph of the construction corresponding to element $e \in \calU$. Here, $f = 8$. The figure in (b) depicts the  intermediate orientation $\vec{H}$ for the subgraph of $H$ corresponding to an element $e \in \calU$. The greyed-out set-vertex at the bottom represents that this vertex is in $X_F$. The figure in (c) depicts the final orientation for the subgraph from the figure in (b) after reorientation. The highlighted edges are those that have been reoriented.}
    \label{fig:gadget-construction}
\end{figure}

    We refer the reader to \Cref{fig:gadget-construction}(a) for pictorial depictions of the instance constructed via the reduction above. The next claim shows the correctness of our reduction and also implies that our reduction is approximation-preserving. The approximation hardness guarantee of the theorem when the target density $\rho = 2$ then follows by \Cref{thm:set-cover-hardness} and the observation that the number of vertices $|V|$ is a constant factor of the size of the input set cover instance.

    \begin{claim}\label{claim:reduction-correctness}
         The instance $(G, c_G)$ has a feasible solution to $2$-\densitydeletionset of finite cost $T$ if and only if $(\calS, \calU, c)$ has a \setcover of cost $T$.
    \end{claim}
    \begin{proof}
        We first show the forward direction of the claim. Let $X \subseteq V$ be a feasible solution to $2$-\densitydeletionset of finite cost $T$. By construction, we have that $X \subseteq \{v_S : S \in \calS\}$. Let $F := \{S : v_S \in X\}$ denote the corresponding sets. By way of contradiction, suppose $F$ is not a set cover. Consequently, there exists an element $e \in \calU$ not covered by $F$. For convenience, we use $\calS(e) := \{S\in \calS:e\in S\}$ to denote the sets that contain the element $e$. We note that since the element $e$ is not covered by $F$, we have that $\{v_S:S\in\calS(e)\}\cap X = \emptyset$. Let $V_e$ denote the set of vertices obtained by including all vertices of the binary tree $\calT_e$. Then, the following gives us a contradiction: 
        $$ 2 \geq \lambda_{G-X}^* \geq \frac{|E[V_e]|}{|V_e|} = \frac{(2f + 1) + (2f-2)}{2f - 1} > 2.$$
        Here, the first inequality is because $X$ is a feasible solution to $2$-\densitydeletionset. The second inequality is by definition of graph density. The equality is because there are $(2f+1)$ self loop edges, $(2f-2)$ non-self loop edges, and $2f-1$ vertices in the tree $\calT_e$ by construction.

        We now show the reverse direction. Let $F \subseteq \calS$ be a set cover. Let $X_F = \{v_S : S \in F\}$. Then, we show that the graph $H := G - X_F$ has density at most $2$. By \Cref{thm:orientation-characterization}(2), it suffices to exhibit an orientation of the graph $H$ in which the indegree of every vertex is at most $2$. We first consider the following intermediate orientation of $G$. For each element $e\in \calU$, we do the following: for vertex $u \in \calT_e - \{r_e\}$, we denote $p(u)$ as the (unique) parent of $u$ in the (rooted) tree, and we orient the edge $(u, p(u))$ towards the parent $p(u)$. All self-loops are assumed to be trivially oriented.
        Let $\vec{H} := (V_H, \vec{E}_H)$ denote this intermediate orientation restricted to the graph $H$, and $\calR := \{r_e: e\in \calU\}$ denote the set of all root vertices.
        Refer to \Cref{fig:gadget-construction}(b)
        for a pictorial depiction of   orientation $\vec{H}$.
        We now make three important observations regarding the indegrees in the orientation $\vec{H}$. 
        \begin{observation}\label{obs:apx-hardness-reduction-helper:orientation-properties}
        We have the following:
        \begin{enumerate}
            \item  for all $u \in V_H - \calR$, $\indegree_{\vec{H}}(u) \leq 2$,
            \item for all $r \in \calR$, $\indegree_{\vec{H}}(r) = 3$, and
            \item for each element $e \in \calU$, there exists a set $S_e \in \calS(e)$ such that $\indegree_{\vec{H}}(p(v_{S_e})) \leq 1$.
        \end{enumerate}
        \end{observation}
        \begin{proof}
            We show each statement separately below.
            \begin{enumerate}
                \item We note that the statement easily follows by construction for the set vertices $v_S$. Let $e \in \calU$ be an arbitrary element and let $u \in \calT_e - \left(\left\{v_S : S \in \calS(e)\right\} \cup \left\{r_e\right\}\right)$ be a non-root internal vertex of the binary tree. Since $u$ has exactly two child nodes and one parent node in $\calT_e$, we have that $\indegree_{\vec{H}}(u) \leq 2$. We note that the inequality may be strict if any children of $u$ belong to the set $X_F$.
                \item Let $e \in \calU$. We note that $r_e$ has exactly $2$ children and $1$ self loop, and consequently $\indegree_{\vec{H}}(r_e) \leq 3$ as claimed. Here, we note that bo child of $r_e$ belongs to the set $X_F$ because of our simplifying assumption that $f \geq 4$.
                \item Let $e \in \calU$. 
                Since $F$ is a set cover, there exists a set $S \in S(e)$ such that $S \in F$. Consequently, $v_S \in X_F$, and so $\indegree_{\vec{H}}(p(v_S)) \leq 1$ by construction.
            \end{enumerate}
        \end{proof}
        We now use the orientation $\vec{H}$ and \Cref{obs:apx-hardness-reduction-helper:orientation-properties} to construct the orientation which certifies that the graph $H$ has density at most $2$. We note that by \Cref{obs:apx-hardness-reduction-helper:orientation-properties}(1) and (2), it suffices to modify the orientation $\vec{H}$ to reduce the indegree of all root vertices in $\calR$ to $2$ while keeping all other indegrees at most $2$. Consider an arbitrary element $e \in \calU$. Let $u_e$ be an arbitrary vertex of $\calT_e$ such that $\indegree_{\vec{H}}(u_e) \leq 1$. We note that such a vertex exists by \Cref{obs:apx-hardness-reduction-helper:orientation-properties}(3). We consider the unique path $P_e$ from $u_e$ to $r_e$ in $\calT_e$. We note that by the construction of the orientation $\vec{H}$, every edge along this path is oriented in the direction of the path. Consider the orientation obtained by reorienting these edges in the reverse direction of the path. Refer to \Cref{fig:gadget-construction}(c)
        for a pictorial depiction of the modified orientation. We note that only the indegrees of $r_e$ and $u_e$ change due to this reorientation. In particular, for all $e \in \calU$, we have that $\indegree_{\vec{H}}(u_e) \leq 2$ and $\indegree_{\vec{H}}(r_e) = 2$. This concludes the proof.
     \end{proof}
    The preceding reduction can be modified to show approximation hardness for all integral $\rho \geq 2$. Suppose that $\rho = 2 + \alpha$, where $\alpha \in \Z_{\geq 0}$. Then, we construct the same graph as above with $\alpha$ additional self-loops on each vertex. The proof generalizes.
\end{proof}

\section{Bicriteria Approximations}\label{sec:bicriteria-algorithm}
In this section, we design bicriteria approximation algorithms for \densitydeletionset and \supmoddensitydeletionset.

\subsection{Bicriteria for \densitydeletionset}\label{sec:orientation-bicriteria}

In this section, we extend the ideas of \cite{chandrasekaran2024polyhedralaspectsfeedbackvertex} to obtain an ILP formulation for \densitydeletionset. We then give a simple threshold-rounding algorithm for its LP relaxation to prove Theorem \ref{thm:orientation-LP-bicriteria}. 

\subsubsection{Orientation ILP}

Our ILP for \rhodensitydeletionset is based on the characterization of $\lambda^*_G$ via fractional orientations given in \Cref{thm:orientation-characterization}. We recall that fractional orientation is an assignment of
two non-negative numbers $z_{e,u}$ and $z_{e,v}$ for each edge $\{u,v\} \in E$ such that 
$z_{e, u}+ z_{e,v}=1$, and
\Cref{thm:orientation-characterization}(1) states that a graph $G$ has density at most $\rho$
if and only if the edges of $G$ can be fractionally oriented such that the total fractional in-degree at every vertex is at most $\rho$.

We describe the details of our formulation now. For an edge $e=uv$ we use variables $z_{e,u}$
and $z_{e,v}$ to denote the fractional amount of $e$ that is oriented towards $u$ and $v$
respectively. Since the \rhodensitydeletionset is a vertex deletion problem, we also have variables
$x_u$ for each $u \in V$ to indicate whether $u$ is deleted. An edge $e=uv$ is in
the residual graph only if $u$ and $v$ are not deleted. These observations
allow us to formulate the ILP for \rhodensitydeletionset below.
\begin{center}
$\begin{array}{ll@{}ll}
        \text{min}  &\displaystyle\sum\limits_{u \in V} c_ux_{u}&   &\\
        \text{s.t.}& x_u + x_v + z_{e, u} + z_{e, v} \geq 1 & & \forall e = uv \in E\\
        &\rho x_u + \displaystyle\sum\limits_{e\in \delta(u)}z_{e, u} \leq \rho &    & \forall u \in V \\
        &z_{e,u} \geq 0&&\forall u \in V, \forall e \in \delta(u)\\
        &x_u \in \{0, 1\}&&\forall u \in V
    \end{array}$ 
\end{center}
We will denote the LP-relaxation of the above ILP for the instance $(G, c, \rho)$ as $\text{LP}_{\text{orient}}(G, c, \rho)$.
\subsubsection{Rounding the Orientation LP}\label{sec:rounding-orientation-LP}

We give our LP-rounding based bicriteria algorithm for \rhodensitydeletionset and prove \Cref{thm:orientation-LP-bicriteria}. In particular, we consider the simple threshold-rounding based algorithm for the orientation LP given in \Cref{alg:bicriteria}. Let $S$ denote the set returned by \Cref{alg:bicriteria}. \Cref{lem:bicriteria-LP-rounding:cost} shows that the cost of the set $S$ is at most $1/\epsilon$ times the cost of an optimum solution and thus satisfies property (2) of the theorem. \Cref{lem:bicriteria-LP-rounding:approximate-feasibility} shows that the density of the graph $G - S$ is at most $\rho/(1-2\epsilon)$ and so the set $S$ satisfies property (1) of the theorem. \Cref{alg:bicriteria}, \Cref{lem:bicriteria-LP-rounding:cost} and \Cref{lem:bicriteria-LP-rounding:approximate-feasibility} together complete the proof of \Cref{thm:orientation-LP-bicriteria}.

\begin{algorithm}
\caption{Bicriteria approximation algorithm for \densitydeletionset}\label{alg:bicriteria}
\textbf{Input:} (1) Graph $G=(V, E)$, (2) Costs $c:V\rightarrow\mathbb{R_{+}}$, (3) Target $\rho \in \Z_+$, (4) Error parameter $\epsilon \in (0, 1/2)$

\begin{enumerate}
    \item Let $(x, z)$ be an optimal solution to $\text{LP}_{\text{orient}}(G, c, \rho)$. 
    \item \textbf{return} $S := \{u \in V : x_u > \epsilon\}$.
\end{enumerate}
\end{algorithm}

\begin{lemma}[Approximate Cost] \label{lem:bicriteria-LP-rounding:cost}
$c(S) \leq \frac{1}{\epsilon}\sum_{u \in V}c_ux_u$.
\end{lemma}
\begin{proof} We have the following.
    $$c(S) = \sum_{u \in S}c_u \leq \frac{1}{\epsilon}\sum_{u \in S}c_ux_u \leq \frac{1}{\epsilon}\sum_{u \in V}c_ux_u,$$
    where the first inequality is by construction of the set $S$.
\end{proof}

\begin{lemma}[Approximate Feasibility] \label{lem:bicriteria-LP-rounding:approximate-feasibility}
    $\lambda^*_{G - S} \leq \rho\cdot \frac{1}{1 - 2\epsilon}$.
\end{lemma}
\begin{proof}
    We will show the claim by exhibiting a fractional orientation of the graph $G-S$ such that the indegree of every vertex is at most $\rho\cdot \frac{1}{1 - 2\epsilon}$. This suffices to prove the lemma via the characterization in  \Cref{thm:orientation-characterization}(1). 
    Let $z'_{e, u} := \frac{1}{1 - 2\epsilon}\cdot z_{e,u}$ for all $e = uv \in E[V - S]$. 
    The first observation below shows that $z'$ is an orientation. The second observation below shows that that the indegree of every vertex is bounded.
    These two observations together show that $z'$ is the desired orientation. 
    \end{proof}
    \begin{claim}
        $z'_{e, u} + z'_{e,v} \geq 1$ for all $e = uv \in E[V - S]$.
    \end{claim}
    \begin{proof}
        Let $e = uv \in E[V - S]$ be arbitrary. Then, we have the following:
        $$1 \leq x_u + x_v + z_{e, u} + z_{e, v} \leq 2\epsilon + z'_{e,u}(1-2\epsilon) + z'_{e,v}(1 - 2\epsilon),$$
        where the first inequality is by the $\text{LP}_{\text{orient}}(G,c,\rho)$ constraint (1) and the second inequality is because $u, v \not \in S$. Then, rearranging the terms gives us the observation.
    \end{proof}

    \begin{claim}
        $\sum_{e \in \delta_{V - S}(u)}z'_{e, u} \leq \rho\cdot\frac{1}{1-2\epsilon}$.
    \end{claim}
    \begin{proof}
        Let $u \in V - S$ be arbitrary. We have the following:
        $$\sum_{e \in \delta_{V - S}(u)}y_{e, u} \leq \sum_{e \in \delta(u)}y_{e, u} \leq \rho(1 - x_u) \leq \rho.$$
            Here, the second inequality is by the $\text{LP}_{\text{orient}}(G,c,\rho)$ constraint (2) and the third inequality is because $x_u$ is non-negative by the $\text{LP}_{\text{orient}}(G,c,\rho)$ constraint (3). The observation then follows because $\epsilon \in (0, 1)$.
    \end{proof}

\subsection{Bicriteria for \supmoddensitydeletionset}\label{sec:supmodDD-bicriteria}

In this section, we describe a randomized combinatorial bicriteria approximation algorithm for $\rho$-\supmoddensitydeletionset and prove Theorem \ref{thm:bicriteria-random-deletion}. 
 The algorithm is inspired by the ideas of the recent work of Włodarczyk \cite{Włodarczyk_2024}.
Our algorithm is based on the following  idea.  Suppose that we had 
non-negative \emph{potentials} $\pi:V\rightarrow\R_{\geq 0}$ for the elements of the ground set   
such that the potential value $\sum_{u\in X}\pi(u)$ of an optimal solution $X \subseteq V$ is large, say at least $\alpha\cdot\sum_{u \in V}\pi(u)$.
Then, a natural algorithm---at least when the vertex deletion costs are uniform---would be to 
compute the potentials, 
sample an element in proportion to the potentials and delete it, define a residual instance, and repeat. 
This would ensure an $\alpha$-approximation for the problem in expectation via a martingale argument.

Unfortunately, our hardness result suggests that we are unlikely to obtain good potentials. In fact, the hard instances seem to be the functions that have density very close to the target density. However, we leverage supermodularity to show that if the density of the input function is at least $\beta$ times the target density, then we can indeed find such good vertex potentials. 
In order to ensure that the density of the input function is at least $\beta$ times the target density, we perform a preprocessing step to prune certain elements from the ground set without changing the cost of an optimal solution. 
Overall, this gives us an $(\alpha, \beta)$-bicriteria guarantee, where the values $\alpha$ and $\beta$ are as given in \Cref{thm:bicriteria-random-deletion}. We also note that the cost function to delete vertices may be arbitrary---we overcome this by using the natural bang-per-buck sampling strategy, i.e. we sample a vertex $u$ in proportion to $\pi(u)/c(u)$.

The rest of the section is organized as follows. 
In \Cref{sec:dense-decomposition} we give our preprocessing step based on the dense decomposition of supermodular functions. 
In \Cref{sec:element-potentials} we describe our element potentials based on marginal gains of supermodular functions. 
In \Cref{sec:bicriteria-random-deletion:algorithm} we present our algorithm and complete the proof of \Cref{thm:bicriteria-random-deletion}.

\subsubsection{Preprocessing via Dense Decomposition}\label{sec:dense-decomposition}
We discuss the \emph{dense decomposition} of a normalized non-negative supermodular set function \cite{hqc-22, Fujishige_1980} and prove a lemma that will enable us to use it as a preprocessing step. 

\begin{definition}\cite{hqc-22}
    Let $f:2^V\rightarrow\R_+$ be a non-negative normalized supermodular function. A sequence $(V_1, \rho_1), (V_2, \rho_2), \ldots, (V_k, \rho_k)$ is the \emph{dense decomposition of $f$} if 
    \begin{enumerate}
        \item $V_1, \ldots, V_k$ is a partition of $V$ obtained iteratively as follows: for $i =1, 2, \ldots, k$, 
        $V_i$ is the inclusion-wise maximal set $S\subseteq V-\cup_{j\in [i-1]}V_j$ that maximizes 
        \[
            \frac{f\left(S \cup \bigcup_{j \in [i-1]}V_j\right) - f\left(\bigcup_{j \in [i-1]}V_j\right)}{|S|}
        \]
   
    \item the values $\rho_1, \ldots, \rho_k$ are obtained as
    $$\rho_i := \frac{f\left(\bigcup_{j \in [i]} V_j\right) - f\left(\bigcup_{j \in [i-1]} V_j\right)}{|V_i|}\ \forall\ i \in [k].$$
    \end{enumerate}
\end{definition}

We note that the dense decomposition can also be viewed algorithmically as the output of a recursive process which computes the unique inclusion-wise maximal set $S$ that maximizes the ratio $f(S)/|S|$ and recurses on the \emph{contracted} function $f_{/S}:2^{V - S}\rightarrow\R$ defined as $f_{/S}(X) := f(S\cup X)$ for all $X \subseteq V - S$. It can be shown that this decomposition is unique for a supermodular $f$.
The following lemma allows us to use the dense decomposition as an algorithmic preprocessing step in the next section. In particular, the lemma says that the dense decomposition can be used to find a set $R\subseteq V$ such that it suffices to focus on solving the \rhosdds{\rho} problem on the function restriction $f|_R$; and additionally, the ground set elements of the restricted function have \emph{large} marginal gains. 

\begin{restatable}{lemma}{lemDenseDecomposition}\label{lem:dense-decomposition}
    Let $f:2^V\rightarrow\R_{\geq 0}$ be a  normalized non-negative supermodular function, $c:V\rightarrow\R_{+}$ be a cost function,  and $\rho \in \R_+$ be a positive real value. Moreover, let $(V_1, \phi_1), (V_2, \phi_2), \ldots, (V_k, \phi_k)$ be the dense decomposition of $(f, V)$ and for $\rho'>\rho$, let $R := \bigcup_{i \in [k] : \phi_i > \rho'}V_i$. Then, we have that 
    \begin{enumerate}
        \item every feasible solution to \rhosdds{\rho'} for the function $f|_R$ is also a feasible solution to \rhosdds{\rho'} for the function  $f$,
        \item $\OPT(f) \geq  \OPT(f|_R)$,
         where $\OPT(f)$and $ \OPT(f|_R)$ denote the costs of optimal solutions to $\rhosdds{\rho}$ for the functions $f$ and $f|_R$ respectively,
        \item $f(v|R-v) \geq \rho'$ for all $v \in R$, and  
        \item the set $R$ can be computed in polynomial time given access to a function evaluation oracle for $f$.
    \end{enumerate}
\end{restatable}
\begin{proof}
    We show all four properties separately below.
    \begin{enumerate}
        \item 
    Let $X \subseteq R$ be a feasible solution to \rhosdds{\rho'} for the function $f|_{R}$ and by way of contradiction suppose that $X$ is not a feasible solution to \rhosdds{\rho'} for the function $f$. Thus, there exists a set $S \subseteq V - X$ such that $f(S)/|S| > \rho'$. We note that this set $S$ cannot be contained in $R - X$ since otherwise $X$ would not be a feasible solution to \rhosdds{\rho'} for $f|_{R}$. Then, the following gives us the required contradiction:
    \begin{align*}
        \rho' < \frac{f(S)}{|S|}  & \leq \frac{f(S\cup R) - f(R) + f(S\cap R)}{|S - R| + |S\cap R|}\\
        & \leq \max\left\{\frac{f(S\cup R) - f(R)}{|S - R|}, \frac{f(S\cap R)}{|S\cap R|}\right\}\\
        &\leq \max\left\{\max\left\{\frac{f(R\cup S') - f(R)}{|S'|} : S'\subseteq V - R\right\}, 
        \lambda_{f|_{R - X}}^*
        )\right\}\\
        & \leq \rho'.
    \end{align*}
    Here, the second inequality is by supermodularity of the function $f$. The third inequality is by the observation that $(a+b)/(c+d) \leq \max\{a/c, b/d\}$ for non-negative numbers $a,b,c,d$. For the final inequality, observe that 
    $\lambda_{f|_{R - X}}^* \leq \rho'$
    because $X$ is a feasible \rhosdds{\rho'} for $f|_{R}$. Furthermore, we have that $\max\left\{\frac{f(R\cup S') - f(R)}{|S'|} : S'\subseteq V - R\right\} \leq \rho'$ by definition of the dense decomposition and $R$.

    \item Let $X\subseteq V$ be an optimal $\rhosdds{\rho}$ for $f$ w.r.t. cost function $c$. Then, we note that $X\cap R$ is a feasible $\rhosdds{\rho}$ for $f|_{R}$. This can be easily observed as follows: by way of contradiction, suppose that $X\cap R$ is not a feasible $\rhosdds{\rho}$ for $f|_{R}$. Then, there exists a set $S \subseteq R - X$ such that $f|_{R}(S)/|S| > \rho$. Consequently, we have that $\lambda^*_{f|_{V - X}} \geq f(S)/|S| =  f|_{R}(S)/|S| > \rho$, a contradiction to $X$ being a feasible \rhosdds{\rho} for $f$. Then, $\OPT(f) \geq \OPT(f|_R)$ follows by non-negativity of $c$. 
        \item By way of contradiction, suppose that there exists a vertex $v \in R$ such that $f(v|R) \leq \rho'$. We recall that $R = \bigcup_{i \in [k] : \phi_i > \rho'}V_i$. Let $j \in [k]$ be such that $v \in V_j$.
        For convenience, we will let $U_{j-1} := \bigcup_{i \in [j-1]}V_i$ and $U_{j} =  \bigcup_{i \in [j]}V_i$. We note that $U_{j-1}$ and $U_j$ are contained in $R$ and $\rho_j := \frac{f(U_j) - f(U_{j-1})}{|V_j|}$ by definition of the dense decomposition. Moreover, $\rho_j > \rho'$ by the definition of the set $R$, and so by supermodularity we have that $f(v|U_j) \leq f(v|R) \leq \rho' < \rho_j$. Then, the following sequence of inequalities gives us the required contradiction.
        \begin{align*}
            \rho_j & = \frac{f(U_j) - f(U_{j-1})}{|V_j|}\\
            & = \frac{f(U_j) - f(U_j - v) + f(U_j - v) -  f(U_{j-1})}{|V_j| }\\
            & = \frac{f(v|U_j - v) + f(U_j - v) -  f(U_{j-1})}{1 + (|V_j| - 1)}\\
            & < \frac{\rho_j + f(U_j - v) -  f(U_{j-1})}{1 + (|V_j| - 1)}\\
            & \leq  \max\left\{\rho_j, \frac{f(U_j - v) -  f(U_{j-1})}{|V_j| - 1}\right\}\\
            & \leq \max\left\{\rho_j, \max\left\{\frac{f(U_{j-1} \cup S) - f(U_{j-1})}{|S|} : S\subseteq V - U_{j-1} \right\}\right\}\\
            & \leq \rho_j,
        \end{align*}
        where the second inequality is by the observation that $(a+b)/(c+d) \leq \max\{a/c, b/d\}$ for non-negative numbers $a,b,c,d$. Here, we note that $|V_j| \geq 2$ since otherwise, $V_j = \{v\}$ and so by supermodularity we have that $\rho_j = f(U_{j-1} + v) - f(U_{j-1}) = f(v |U_{j-1}) \leq f(v|R) < \rho_j$, 
        a contradiction.
        \item It is well-known that the dense decomposition (and consequently the set $R$) can be computed in polynomial time (given access to the function evaluation oracle for  $f$) via supermodular maximization. This is implicit in Fujishige's work on principle partitions \cite{Fujishige_1980}, and is also explicitly considered in more recent works on dense decompositions for supermodular functions \cite{hqc-22}. We omit the details of a formal proof here for brevity. 
    \end{enumerate}
\end{proof}


\subsubsection{Element Potentials via Marginal Gains}\label{sec:element-potentials}
We now show that the marginal gains of elements relative to the entire ground set are good potentials for a sampling-based algorithm for $\rho$-\supmoddensitydeletionset when all the marginal gains of the input function are large enough, in particular, at least $c_f(1+\epsilon)\rho$. 

\begin{restatable}{lemma}{lemBicritHelper}\label{lem:bicriteria-random-deletion-helper:main}
Let $\rho\in \R$ and $\epsilon \in (0, 1)$. Let $f:2^V\rightarrow\Z$ be a normalized monotone supermodular function such that $f(u|V - u) \geq c_f(1+\epsilon)\rho$ for all $u \in V$, and let $X \subseteq V$ be a \rhosdds{\rho} for $f$. Then, we have that
$$\sum_{u \in X} f(u|V-u) \geq \frac{1}{c_f (1 + 1/\epsilon)}\sum_{u \in V}f(u|V-u).$$
\end{restatable}
\begin{proof}
By supermodularity of $f$, we have that 
\begin{align}
\sum_{u\in X}f(V-u)&\le (|X|-1)f(V) + f(V-X) \text{ and hence, }\notag\\
\sum_{u\in X}f(u|V-u) &\ge f(V) -f(V-X). \label{eq:sum-of-opt-marginals}
\end{align}
By way of contradiction, suppose that the lemma is false. Then, we have the following.
    \begin{align*}
        \sum_{u \in V}f(u | V-u)& \leq c_{f} f(V)&\\
        &= c_{f}(f(V) - f(V - X) + f(V - X))&\\
        &\leq c_f\left(f(V) - f(V - X)\right) + c_f\rho|V - X|&\\
        &\leq c_f\sum_{u \in X}f(u|V-u) + c_f\rho|V - X|&\\
        &< \frac{1}{1 + 1/\epsilon}\sum_{u \in V}f(u|V-u) + c_f \rho|V - X|.&
    \end{align*}
    Here, the first inequality is by definition of the parameter $c_f$. The second inequality is because $X$ is a feasible \rhosdds{\rho} for the function $f$. The third inequality is by \eqref{eq:sum-of-opt-marginals}. 
    The final inequality is by our contradiction assumption.
    Then, on rearranging the terms, we obtain the following contradiction.
    $$c_f\rho|V - X| > \left(1 - \frac{1}{1+ 1/\epsilon}\right)\sum_{u \in V}f(u|V-u) = \frac{1}{(1+\epsilon)}\sum_{u \in V}f(u|V-u) \geq c_f\rho|V|,$$
   where the final inequality is because $f(u|V-u) \geq c_f(1+\epsilon)\rho$ for all $u \in V$.
   \end{proof}

\subsubsection{Random Deletion Algorithm}\label{sec:bicriteria-random-deletion:algorithm}
We now describe our bicriteria algorithm for $\rho$-\supmoddensitydeletionset and analyze its approximation factor. \Cref{alg:bicriteria-iterative-random-deletion}, \Cref{lem:bicriteria-random-deletion:approximate-feasibility} and \Cref{lem:bicriteria-random-deletion:expected-cost}
together complete the proof of \Cref{thm:bicriteria-random-deletion}.

\paragraph{Algorithm.} Our algorithm takes as input (1) a normalized non-negative supermodular function $f:2^V\rightarrow\R_+$, (2) element deletion costs $c:V\rightarrow\R_+$, (3) target density $\rho \in \R_+$, and (4) error parameter $\epsilon > 0$. The algorithm returns a set $S \subseteq V$ which starts off as the empty-set and is then constructed element-by-element. This is done iteratively as follows. Let $\beta:=c_f (1+\epsilon)$. If the function $f$ has density at most $\beta\rho$, then the algorithm breaks and returns the current set $S$. Otherwise, the algorithm first computes the dense decomposition $(V_1, \phi_1), (V_2, \phi_2), \ldots, (V_k, \phi_k)$ of  the function $f$, defines the set $R := \cup_{i \in [k] : \phi_i > \beta\rho}V_i$, and redefines the function $f$ to be the restricted function $f|_R:2^R \rightarrow\R_{\geq 0}$---we use \textsc{DenseDecompositionPreprocess}$(f, \rho)$ to denote a subroutine that computes the set $R$ and returns the tuple $(f|_{R}, R)$. Next, the algorithms samples a random element $u$ from the (modified) set $V$ in proportion to the ratio $f(u|V-u)/c(u)$. The algorithm then adds the vertex $u$ to the set $S$, restricts $f$ to the ground set $V - u$, and repeats the previous steps. We give a formal description of the algorithm in \Cref{alg:bicriteria-iterative-random-deletion}.

\begin{algorithm}
\caption{Bicriteria approximation algorithm for $\rho$-\supmoddensitydeletionset}\label{alg:bicriteria-iterative-random-deletion}
\textsc{Algorithm}$\left(\left(f:2^V\rightarrow\R ,c\right), \rho, \epsilon\right)$:
\begin{enumerate}
\item $S := \emptyset$
\item \textbf{while} $\lambda_f^* > c_f (1+\epsilon)\rho$:
\begin{enumerate}[nosep]
        \item Redefine $(f, V) := $ \textsc{DenseDecompositionPreprocess}$(f, c_f (1+\epsilon)\rho)$
        \item $u :=$ vertex sampled from $V$ according to the following distribution: $$\Pr(u = v) := \frac{f(v|V-v)}{c(v)\cdot W} \ \forall v \in V, \text{ where $W := \sum_{v \in V}\frac{f(v|V-v)}{c(v)}$ is a normalizing factor}$$
        \item $S := S+u$ and $f := f|_{V - u}$
    \end{enumerate}
    \item \textbf{return} $S$. 
\end{enumerate}
\end{algorithm}

\paragraph{Martingales. } For the analysis of our randomized algorithm, we will require the following concepts from probability theory.
\begin{definition} 
\begin{enumerate}
    \item 
    A sequence of random variables $P_1, P_2, ...$ is called a \emph{supermartingale} w.r.t. the sequence $X_1, X_2, \ldots $ of random variables if for each $i \in \Z_+$ it holds that (i) $P_i$ is a function of $X_1, \ldots X_i$,
    (ii) $\E[|P_i|] < \infty$  and (iii) $\E[P_{i+1} | X_1, \ldots X_i] \leq P_i$. 
    \item A random variable $T$ is called a \emph{stopping time} with respect to the sequence of random variables $P_1, P_2, ...$  if for each $i \in \Z_+$, the event $(T \leq i)$ depends only on $P_1, \ldots, P_i$.
    \end{enumerate}
\end{definition}

The following result shows that the expected value of a random variable in the supermartingale process only decreases with time. This will be crucial in analyzing the performance of \Cref{alg:bicriteria-iterative-random-deletion}.

\begin{theorem}[Doob’s Optional-Stopping Theorem]\label{thm:doob}
    Let $P_0, P_1, \ldots $ be a supermartingale  w.r.t. the sequence $X_1, X_2, \ldots $ of random variables and $\ell$ be a stopping time with respect to the process $P$. Suppose that $\Pr(\ell \leq n) = 1$ for some integer $n\in \Z_+$. Then, we have that  $\E[P_\ell] \leq \E[P_0]$.
\end{theorem} 

\paragraph{Algorithm Analysis. } Henceforth, we consider the execution of \Cref{alg:bicriteria-iterative-random-deletion} on a fixed input instance $(f, c, \rho, \epsilon)$.  Let $\ell \in \Z_+$ be the number of iterations of the while-loop---we note that $\ell$ is a random variable with value at most $n$ since at every iteration of the while-loop, the size of the ground set decreases by at least $1$. Throughout the analysis, we will index the (random) variables at the $i^{th}$ iteration of the algorithm with the subscript $i$ for all $i \in [\ell]$. In particular, we let $S_i$ denote the set $S$ at the start of the $i^{th}$ iteration (so $S_1 := \emptyset$, and $S_{i+1}$ is defined by Step 2(c)), $f_i:2^{V_i}\rightarrow\R_{\geq0}$ denote the preprocessed function $f$ after step 2(a), and $u_i$ denote the sampled vertex $u$ after step 2(b) of the $i^{th}$ iteration of the algorithm.
For simplicity, we define $S_{j} := S$, and $f_{j}$ to be the empty-function for all $j \geq \ell$. The next lemma shows that the density of the function after deleting the set $S$ is at most $c_f(1+\epsilon)\rho$, i.e. $S$ is a feasible solution to \rhosdds{(c_f(1+\epsilon)\rho)} for the function $f$. The proof easily follows by considering any fixed execution of the algorithm and leveraging \Cref{lem:dense-decomposition}(1) while inducting on $\ell$. We omit details of the proof here for brevity.

\begin{lemma}[Approximate Feasibility]\label{lem:bicriteria-random-deletion:approximate-feasibility}
    $\lambda^*_{f|_{V - S}} \leq c_f(1+\epsilon)\rho$.
\end{lemma}

The next lemma shows that the expected cost of the solution returned by the algorithm is at most $c_f(1+1/\epsilon)\rho$ times the cost of the optimal \rhosdds{\rho} of the function $f$. For any restriction $g$ of the function $f$, we use $\OPT(g)$ to denote the value of an optimal $\rhosdds{\rho}$ for $g$ with respect to the cost function $c$.

\begin{lemma}[Approximate Cost]\label{lem:bicriteria-random-deletion:expected-cost}
    $\E[c(S)] \leq c_f(1+1/\epsilon)\OPT(f)$.
\end{lemma}
\begin{proof}
    For ease of exposition, we will use $\alpha:= c_f(1+1/\epsilon)$.
    We consider the sequence of random variables $P_1, P_2, \ldots$, where $P_i := c(S_i) + \alpha\OPT(f_i)$ for all $i \in \Z_+$. Our strategy will be to first show that this sequence of random variables is a supermartingale, and then apply Doob's Optional-Stopping Theorem with stopping time $n$ to bound the expected cost of the set returned by the algorithm (note that $\ell \leq n$ since with each iteration of the while-loop, the size of the ground set decreases by at least $1$). Before showing that the sequence is a supermartingale, we first show that the expected cost of a vertex chosen in step 2(b) of an iteration of the algorithm is at most an $\alpha$-fraction of the expected decrease in the optimum value during the iteration.

    \begin{claim}\label{claim:supermartingale-helper} 
    $\E[c(u_i)|u_1, u_2, \ldots, u_{i-1}] \leq \alpha\E\left[\OPT(f_i) - \OPT(f_{i+1}) |u_1, u_2, \ldots, u_{i-1}\right]$ for all $i \in [\ell]$.
\end{claim}
\begin{proof}
Let $X_i$ be an optimal \rhosdds{\rho} for $f_i$. 
We have the following:
\begin{align*}
    \E[c(u_i) | u_1, u_2, \ldots, u_{i-1}] 
    &= \sum_{v \in V_i} \Pr(u_i = v)\cdot c(v) \\
    &= \frac{1}{W}\sum_{v \in V_i}f_i(v|V_i)\\
    &\leq \frac{\alpha}{W}\sum_{v \in X_i} f_i(v|V_i)\\
    &= \alpha\sum_{v \in X_i} \Pr(u_i = v)\cdot c(v),
 \end{align*}
 where the first inequality is by \Cref{lem:bicriteria-random-deletion-helper:main} and the fact that $f(v|V_i) \geq c_f(1+\epsilon)\rho$ by our preprocessing (Step 2(a)) and \Cref{lem:dense-decomposition}(3). 
 We now show that because $X_i$ is an optimal solution for $f_i$, the final expression in the above can be upper bounded by $\alpha\E[\OPT(f_i) - \OPT(f_{i+1})]$, thereby completing the proof of the claim. 
 This can be seen as follows:
\begin{align*}
    \E[\OPT(f_i) - \OPT(f_{i+1})] & = \sum_{v \in V_i}\E[\OPT(f_i) - \OPT(f_{i+1}) | u_i = v]\cdot \Pr(u_i = v)\\
    & \geq \sum_{v \in X_i} \E[\OPT(f_i) - \OPT(f_{i+1}) | u_i = v]\cdot \Pr(u_i = v)\\
    &\geq \sum_{v \in X_i} \E[\OPT(f_i) - \OPT(f_{i}|_{V_i - v}) | u_i = v]\cdot \Pr(u_i = v)\\
    & \geq  \sum_{v \in X_i} \E[c(X_i) - c(X_i - v) | u_i = v]\cdot \Pr(u_i = v)\\
    & = \sum_{v \in X_i}c(v)\cdot \Pr(u_i = v).
\end{align*}
Here, the second inequality is by Step 2(b) of \Cref{alg:bicriteria-iterative-random-deletion} which says that the function $f_{i+1}$ is defined to be $ \textsc{DenseDecompositionPreprocess}(f_i|_{V_i - v}, c_f(1+\epsilon)\rho)$ and \Cref{lem:dense-decomposition}(2). The third inequality is because $X_i$ is an optimal solution for $f_i$ and $X_i - v$ is a feasible solution for $f_i|_{V_i - v}$.
\end{proof}
    
    We now show that the sequence $P_1, P_2, \ldots$ is a supermartingale w.r.t. the sequence of random variables $u_1, u_2, \ldots$ chosen by the algorithm.
    \begin{claim}\label{claim:supermartingale:main}
        The sequence of random variables $P_1, P_2, \ldots$ is a supermartingale  w.r.t. the sequence of random variables $u_1, u_2, \ldots$.
    \end{claim}
    \begin{proof}
        Let $i\in \Z_+$ be arbitrary. We note that $P_i$ has finite expectation and also is fully determined by the subsequence $u_1, \ldots, u_i$. Thus, our goal is to show that $\E[P_{i+1} | u_1, \ldots, u_i] \leq P_i$. This is equivalent to showing $\E[P_{i+1} - P_i | u_1, \ldots, u_{i}] \leq 0$. We note that this inequality indeed holds because
        $$\E[P_{i+1} - P_i | u_1, \ldots, u_{i}] = \E\left[c(u_{i+1}) - \alpha\left(\OPT(f_i) - \OPT(f_{i+1})\right) |  u_1, \ldots, u_i\right] \leq 0,$$
        where the inequality is by Claim \ref{claim:supermartingale-helper}.
    \end{proof}

    By Claim \ref{claim:supermartingale:main}, the sequence $P_1, P_2, \ldots$ is a supermartingale w.r.t. the sequence $u_1, u_2, \ldots$ of random variables. Consider the stopping time $\ell$. By \Cref{thm:doob}, we have that $\E[P_\ell] \leq \E[P_1]$. The following then completes the proof of the lemma:
    $$\E[c(S)] = \E[c(S_{\ell}) + \alpha\OPT(f_{\ell})] = \E[P_{\ell}] \leq \E[P_1] = c(S_1) + \alpha\OPT(f_1) \leq \alpha\OPT(f).$$
    Here, the final inequality follows by observing that $f_1 = \textsc{DenseDecompositionPreprocess}(f, c_f(1+\epsilon)\rho)$ and applying \Cref{lem:dense-decomposition}(2).
\end{proof} 
\section{\submodcover and \supmoddensitydeletionset}\label{sec:submodcover-reductions}

In this section, we prove Theorems \ref{thm:sddsToSubmodCover} and \ref{thm:SubmodCovertosdds}. 

\thmsddsToSubmodCover*
\begin{proof}
For simplicity, we define an intermediate function $g:2^V\rightarrow\R_{\geq 0}$, and use it to define the function $h:2^V\rightarrow\R_{\geq 0}$ of interest. The functions $g$ and $h$ are as follows: for every $X \subseteq V$,
\begin{align*}
    g(X) &:= \max\{f(Z) - \rho|Z| : Z \subseteq X\}, \text{ and }\\
    h(X) & := g(V) - g(V - X).
\end{align*}

We note that the function $h$ is normalized, non-decreasing, and submodular. If $\rho$ is an integer, then $h$ is integer-valued. Moreover, we  can answer evaluation queries for $h$ using polynomial many evaluation queries to $f$ (via supermodular maximization). 
We prove properties (2) and (3) of the theorem below.
\begin{enumerate}
    \item[(2)] Let $F\subseteq V$. We have the following sequence of equivalences.
    \begin{align*}
        \lambda^*_{f_{V - F}} \leq \rho \Leftrightarrow\  & \max_{S\subseteq V - F}\left\{\frac{f(S)}{|S|}\right\} \leq \rho\\
        \Leftrightarrow\ & g(V - F) \leq 0 \\
        \Leftrightarrow\ & g(V) - g(V - F) - g(V) + g(\emptyset) \geq 0\\
    \Leftrightarrow\ & h(F) - h(V) \geq 0
    \end{align*}
    Here, the second equivalence can be seen by the following. 
    For the forward direction, we suppose that $ \max\left\{\frac{f|_{V - F}(S)}{|S|} : S\subseteq V - F\right\} \leq \rho$. By way of contradiction, suppose that $g(V - F) > 0$. By definition of the function $g$, there exists a set $Z^* \subseteq V-F$ such that $g(V - F) = f(Z^*) - \rho|Z^*|$. Thus, $f(Z^*) - \rho|Z^*| > 0$. Equivalently, we have that $f(Z^*)/|Z^*| > \rho$, a contradiction to our hypothesis. Here we note that $Z^* \not = \emptyset$ as otherwise we would have that $f(Z^*) - \rho|Z^*| = 0$ since our function $f$ is normalized, contradicting our choice of $Z^*$. For the reverse direction, suppose that $g(V - F) \leq 0$. By way of contradiction, suppose that there exists a non-empty set $S^* \subseteq V - F$ such that $f(S^*)/|S^*| > \rho$. Then, we equivalently have that $f(S^*) - \rho|S^*| > 0$, a contradiction.
    \item[(3)] Let $v \in V$. We have the following:
    \begin{align*}
        h(v)& = g(V) - g(V - v)\\
        &= \max\left\{f(Z) - \rho|Z|: Z \subseteq V \right\} - \max \left\{f(Z') - \rho|Z'| :  Z' \subseteq V - v\right\} \\
        &= \max \left\{0, \max\left\{f(Z) - \rho|Z|: v \in Z \subseteq V \right\} - \max \left\{f(Z') - \rho|Z'| :  Z' \subseteq V - v\right\}\right\} \\
        & \leq \max\left\{0 , \max\left\{(f(Z) - \rho|Z|) - (f(Z') - \rho|Z'|) : v \in Z\subseteq V \text{ and } Z' \subseteq V - v\right\}\right\}\\
        & \leq \max\left\{0 , \max\left\{(f(Z) - \rho|Z|) - (f(Z-v) - \rho|Z-v|) : v \in Z\subseteq V\right\}\right\} \\
        & = \max\left\{0 , \max\left\{(f(Z) - f(Z-v) : v \in Z\subseteq V \right\} - \rho\right\}\\
        & \leq \max\left\{0, f(V) - f(V - v) - \rho\right\},
    \end{align*}
    where the final inequality is because $f(V - v) + f(Z) \leq f(V) + f(Z - v)$ for all $Z \subseteq V$ such that $v \in Z$ by supermodularity of the function $f$.
\end{enumerate}
\end{proof}

\thmSubmodCovertosdds*
\begin{proof}
    We consider the function $f:2^V\rightarrow\Z_{\geq 0}$ defined as follows: 
    for every $X \subseteq V$,
$$ f(X) := h(V) - h(V - X) + |X|.$$
We note that the function $f$ is normalized, non-decreasing, integer-valued and supermodular. Moreover, we  can answer evaluation queries for the function $f$ using two queries to the evaluation oracle for $h$. 
We note that property (1) of the theorem can be observed by following the steps of the proof of \Cref{thm:sddsToSubmodCover}(2) in reverse order, and so we omit the formal details here for brevity. Property (2) of the theorem can be observed as follows: for $v \in V$, we have the following:
\begin{align*}
f(v|V - v) = f(V) - f(V - v) &= (h(V) - h(\emptyset) + |V|) - (h(V) - h(v) + |V - v|)\\
&= h(v) + 1,
\end{align*}
where the final equality is because $h$ is normalized.
\end{proof}
\section{Conclusion}\label{sec:conclusion}
In this work, we considered several interrelated density deletion problems motivated by the question of understanding the robustness of densest subgraph. We showed tight logarithmic approximations for these problems. We showed inapproximability of graph density deletion by reduction from set cover and approximation algorithms by exhibiting the equivalence of supermodular density deletion and submodular cover. Motivated by our hardness results, we designed bicriteria approximation. Our bicriteria approximation for graph density deletion is LP-based and that for supermodular density deletion is randomized, combinatorial, and relies on the notion of dense decomposition of supermodular functions. 
We mention two open questions raised by our work. Firstly, we note that our bicriteria approximation for supermodular density deletion depends on the parameter $c_f$ related to the input supermodular function (see Theorem \ref{thm:bicriteria-random-deletion}). Is it possible to design a bicriteria approximation without the dependence on the parameter $c_f$? Secondly, we note that our hardness reduction shows that $\rho$-\dds is $\Omega(\log{n})$-hard for every fixed constant integer $\rho\ge 2$. We were able to adapt our reduction to conclude that it is $\Omega(\log{n})$-hard for every fixed constant $\rho\ge 3$ (not necessarily integers). Is it $\Omega(\log{n})$-hard for every fixed constant $\rho>1$?

\paragraph{Acknowledgments.}
 Shubhang Kulkarni would like to thank Kishen Gowda for engaging in preliminary discussions about approximations for feedback vertex set and pseudoforest deletion set. These turned out to be a starting point for us in realizing the connections to submodular cover implicitly known in literature. We thank an anonymous reviewer for pointing out an error in our proof of a previous version of  \Cref{thm:DDS-logn-hard}.

\bibliographystyle{abbrv}
\bibliography{references}

\appendix
\section{Reductions Between \sdds, \mfvs, and \dds}\label{appendix:sec:reductions}

In this section, we show reductions between \sdds, \mfvs and \dds. Throughout the section, we use $b_{G}(Z) := \cup_{u \in Z}\delta_G(u)$ to denote the \emph{edge-coverage} of  a set of vertices $Z \subseteq V$ in a graph $G = (V, E)$. Furthermore, for a matroid $\calM = (E, \calI)$, we use $\matroidrank_{\calM}$ to denote its rank function, and $\calM^*$ to denote the dual matroid. We first show that \mfvs is a special case of \sdds.
\begin{theorem}\label{thm:matroidfvs-to-supmodDD}
    Let $G = (V, E)$ be a graph and $\calM = (E, \calI)$ be a matroid. Then, there exists a normalized, non-negative, integer-valued, supermodular function $f:2^V\rightarrow\Z_{\geq 0}$ such that for a subset $F\subseteq V$ of vertices, we have that $E[V - F] \in \calI$ if and only if $\lambda^*_{f|_{V - F}}\leq 1$.
\end{theorem}
\begin{proof}
We prove the theorem in two steps. For the first step, we construct an intermediate function $h:2^V\rightarrow\Z_{\geq 0}$ defined as follows: $h(S) := \matroidrank_{\calM^*}(b_G(S))$ for all $S \subseteq V$. The following claim, implicit in \cite{Fujito-matroid-fvs}, shows that this construction reduces checking independence in the matroid $\calM$ to Submodular Cover.
\begin{claim}[\cite{Fujito-matroid-fvs}]
    The function $h$ is integer-valued, normalized, non-decreasing and submodular. Moreover, for all $F \subseteq V$, we have that $E[V - F] \in \calI$ if and only if $h(F) = h(V)$.
\end{claim}
\begin{proof}
    We note that since the functions $b_G$ and $\matroidrank_{\calM^*}$ are integer-valued, normalized, non-decreasing and submodular, these properties also hold for the function $h$. Furthermore, the second part of the claim from the definition of $h$ and the following two properties of matroids: $\matroidrank_{\calM^*}(E') = |E'| - \matroidrank_{\calM}(E) + \matroidrank_{\calM}(E-E')$ for all $E' \subseteq E$, and $E[V - F] \in \calI$ if and only $\matroidrank_{\calM}(E[V - F]) = |E[V - F]|$.
\end{proof}

For the second step, we use the above intermediate function $h$ to obtain the required normalized supermodular function $f:2^V\rightarrow\R_{\geq 0}$. We note that this function $f$ can be constructed by applying the same construction as in the proof of \Cref{thm:SubmodCovertosdds} and observing the additional two properties of non-negativity and non-decreasing monotonicity. Delving into the proof of \Cref{thm:SubmodCovertosdds}, we observe that the function $f$ is explicitly defined as follows: $f(S) := h(V) - h(V - S) + |S|$ for all $S \subseteq V$. We omit repeating the details here for brevity.
\end{proof}

Next, we show that for all integer $\rho \in \Z_+$, \rhodensitydeletionset is a special case of \mfvs. For this, we will need the following background. We recall that for an integer $\rho\in\Z_+$, the $\rho$-fold union of a matroid $\calM = (E, \calI_\rho)$ is another matroid $\calM_\rho = (E, \calI_{\rho})$, where a subset of edges $F\subseteq E$ is in $\calI_{\rho}$ if $F$ can be partitioned into $\rho$ parts such that each part is in $\calI$, i.e., $F := \uplus_{i \in [\rho]} F^{(i)}$ such that $F^{(i)} \in \calI$ for every $i\in [\rho]$. We refer the reader to Welsh's book on matroid theory \cite{Welsh-book} for additional details. 
We will also rely on a well-known characterization of \emph{pseudoforests} using density---we recall that a graph is a pseudoforest if every component has at most one cycle. The following proposition states that pseudoforests are the graphs that have density at most $1$.
\begin{proposition}\cite{chandrasekaran2024polyhedralaspectsfeedbackvertex}\label{prop:psuedoforest-characterization}
    A graph $G$ is a pseudoforest if and only if $\lambda^*_G \leq 1$.
\end{proposition}

We now show the connection between \rhodensitydeletionset and \mfvs when $\rho \in \Z_+$.

\begin{theorem}\label{thm:dds-to-matroidfvs}
    Let $G = (V, E)$ be a graph and $\rho \in \Z_+$ be an integer. 
    Let $\calM := (E, \calI)$ denote the pseudoforest matroid on the graph $G = (V, E)$, where a set of edges $E'\subseteq E$ is independent if every component in the subgraph $G' = (V, E')$ has at most one cycle. Let $\calM_{\rho} := (E, \calI_\rho)$ be the $\rho$-fold union of the pseudoforest matroid. 
    Then, for a subset $F \subseteq V$, we have that $E[V - F] \in \calI_\rho$ if and only if $\lambda^*_{G - F} \leq \rho$.
\end{theorem}
\begin{proof}
Let $F \subseteq V$ be an arbitrary subset of vertices. 
For notational convenience, we denote $G-F$ by $G_F = (V_F, E_F)$.

For the reverse direction, suppose that $\lambda^*_{G_F} \leq \rho$. Consequently, by \Cref{thm:orientation-characterization}(2), there exists an orientation $\vec{G}_F$ of the graph $G_F$ such that $\indegree_{\vec{G}_F}(u) \leq \rho$ for all $u \in V_F$. Using this orientation, we partition the edges $E_F$ into $\rho$ parts $E_F^{(1)}, \ldots, E_F^{(\rho)}$ such that for the graph $G_F^{(i)} := (V_F, E_F^{(i)})$ where $i \in [\rho]$, we have that $\indegree_{G_F^{(i)}}(u) \leq 1$ for all $u \in V$. We note that this partitioning can be obtained by the following simple iterative procedure for $\rho$ iterations: during the $i^{th}$ iteration, we construct the graph $G_F^{(i)}$ by letting every vertex pick an unpicked edge oriented into the vertex (if there is such an edge).  Then, by \Cref{thm:orientation-characterization}(2) and \Cref{prop:psuedoforest-characterization}, the graph $G_F^{(i)}$ is a pseudoforest, and so $E_F^{(i)} \in \calI$ for all $i \in [\rho]$ by definition of the pseudoforest matroid $\calM$. Consequently, we have that $E_F \in \calI_\rho$ by definition of the $\rho$-fold union matroid.

For the forward direction, suppose that $E_F \in \calI_\rho$. We will perform the argument of the forward direction in reverse. By definition of the $\rho$-fold union matroid, we can partition the edges of $E_F$ into $\rho$ parts $E_F^{(1)}, \ldots, E_F^{(\rho)}$ such that the graph $G_F^{(i)} := (V_F, E_F^{(i)})$ is a pseudoforest for each $i \in [\rho]$. Then, by \Cref{thm:orientation-characterization}(2) and \Cref{prop:psuedoforest-characterization}, we can independently obtain an orientation $\vec{G}_F^{(i)}$ such that $\indegree_{\vec{G}_F^{(i)}}(u) \leq 1$ for all $i \in [\rho]$. Composing all these orientation together, we obtain an orientation $\vec{G}_F$ of the graph $G_F$ such that $\indegree_{\vec{G}_F}(u) \leq \rho$. Then, by \Cref{thm:orientation-characterization}(1), we have that $\lambda^*_{G_F} \leq \rho$.
\end{proof}

\begin{remark}
    There are several ways to prove \Cref{thm:dds-to-matroidfvs}. For example, one can directly show that $\calM = (E, \calI := \{E'\subseteq E : \lambda^*_{(V, E')} \leq \rho\})$ is a matroid by showing that it satisfies the matroid axioms or by
    applying known results in submodularity and matroid theory (see Corollary 8.1 of \cite{Welsh-book}).
\end{remark}

\end{document}